\documentclass[11pt,a4paper]{article}

\usepackage[english]{babel}
\usepackage[T1]{fontenc}
\usepackage{xspace}
\usepackage{mathrsfs}
\usepackage{amsfonts}
\usepackage{amssymb}
\usepackage{amsmath}
\usepackage{amsthm}
\usepackage{mathbbol}

\usepackage{enumitem}

\usepackage{mathabx}

\usepackage{framed}

\def\bN{\mathbb{N}}

\def\cF{\mathcal{F}}
\def\cG{\mathcal{G}}
\def\cH{\mathcal{H}}

\def\cP{\mathcal{P}}
\def\cS{\mathcal{S}}

\def\cSG{\mathcal{SG}\xspace}
\def\cGG{\mathcal{GG}\xspace}

\def\Red{{\bf Red}\xspace}
\def\Blue{{\bf Blue}\xspace}
\def\wrt{\emph{w.r.t.}\xspace}
\def\ie{\emph{i.e.}\xspace}

\def\comp#1{\widebar{{#1}}}

\def\wrt{\emph{w.r.t.}\xspace}
\def\ie{\emph{i.e.}\xspace}

\newcommand{\mim}{{\mathbf{mim}}}

\newcommand{\lmimw}{{\mathbf{lmimw}}}

\newcommand{\fr}{\mathbf{frac}}

\newcommand{\trace}{{\mathsf{trace}}}
\newcommand{\vertex}{{\mathsf{vert}}}

\newtheorem{thm}{Theorem}
\newtheorem{lem}[thm]{Lemma}

\newtheorem{cor}[thm]{Corollary}
\newtheorem{prop}[thm]{Proposition}

\newtheorem{fact}[thm]{Fact}

\begin{document}

\title{Output-Polynomial Enumeration on Graphs of Bounded (Local) Linear 
MIM-Width\footnote{The research leading to these results has received funding from the European Research Council under the European Union's Seventh Framework Programme (FP/2007-2013) / ERC Grant Agreement n. 267959.}}

\author{
Petr A. Golovach\thanks{Department of Informatics, University of Bergen, Norway}
\addtocounter{footnote}{-1} 
\and
Pinar Heggernes\footnotemark
\and
Mamadou Moustapha  Kant\'e\thanks{Clermont-Universit{\' e}, Universit{\'e} Blaise Pascal, LIMOS, CNRS, Aubi{\' e}re, France}
\and
Dieter Kratsch\thanks{Universit\'e de Lorraine, LITA, Metz, France}
\addtocounter{footnote}{-3} 
\and
Sigve H. S{\ae}ther\footnotemark
\addtocounter{footnote}{-1}
\and 
Yngve Villanger\footnotemark
}

\maketitle

\begin{abstract} 
The linear induced matching width (LMIM-width) of a graph is a width parameter defined by using the notion of branch-decompositions of a set function on ternary trees. In this paper we study output-polynomial enumeration algorithms on graphs of bounded LMIM-width and graphs of bounded local LMIM-width. In particular, we show that all 
1-minimal and all 1-maximal $(\sigma,\rho)$-dominating sets, and hence all minimal dominating sets, of graphs of bounded LMIM-width can be enumerated with polynomial (linear) delay using polynomial space. Furthermore, we show that all minimal dominating sets of a unit square graph can be enumerated in incremental polynomial time. 
\end{abstract}

\section{Introduction}
Enumeration is at the heart of computer science and combinatorics. Enumeration algorithms for graphs and hypergraphs typically deal with listing all vertex subsets or edge subsets satisfying a given property. As the size of the output is often exponential in the size of the input, it is customary to measure the running time of enumeration algorithms in the size of the input
plus the size of the output. If the running time of an algorithm is bounded by a polynomial in the size of the input plus the size of the output, then the algorithm is called output-polynomial. 
A large number of such algorithms 
have been given over the last 30 years; many of them solving problems on graphs and 
hypergraphs ~\cite{EiterG95,EiterG02,EiterGM03,JohnsonP88,KhachyanBBEG08,%
KhachyanBEG08,LawlerL80,Tarjan73}. It is also possible to show that certain enumeration problems have no output-polynomial time algorithm unless P\,=\,NP \cite{KhachyanBBEG08,KhachyanBEG08,LawlerL80}. 

Recently Kant{\'e} et al. showed that the famous longstanding open question whether there is an output-polynomial algorithm to  enumerate all minimal transversals of a hypergraph, is equivalent 
to the question whether there is an output-polynomial algorithm to enumerate all minimal dominating sets of a graph~\cite{KanteLMN12}. 
Although the main question remains open, a large number of results have been obtained on graph classes.
Output-polynomial algorithms to enumerate all minimal dominating sets
exist for graphs of bounded treewidth and of bounded clique-width \cite{Courcelle09}, interval graphs \cite{EiterG95}, strongly chordal graphs \cite{EiterG95},  planar graphs
\cite{EiterGM03}, degenerate graphs \cite{EiterGM03},  split graphs \cite{KanteLMN12}, path graphs~\cite{KanteLMN12b}, permutation graphs~\cite{KanteLMNU13}, line
graphs~\cite{GolovachHKV14,KanteLMN12b,KanteLMNU14}, chordal bipartite 
graphs~\cite{GolovachHKKV15}, chordal graphs~\cite{KanteLMNU14a} and graphs of girth at least 7~\cite{GolovachHKV14}.

In this paper, we extend the above results to graphs of bounded linear maximum induced matching width.  Using the notion of branch-decomposi\-tions of a set function on ternary trees introduced by
Robertson and Seymour, the notion of \emph{maximum induced matching width} (MIM-width) was introduced by Vatshelle~\cite{Vatshelle12}.  The linear maximum induced matching width (LMIM-width) of a
graph is the linearized variant of the MIM-width like path-width is the linearized version of tree-width.  (For definitions, see Section~\ref{sec:defs}.)
Belmonte and Vatshelle showed that several important graph classes, among them interval, circular-arc and permutation graphs, have bounded LMIM-width~\cite{BelmonteV13}. Polynomial-time algorithms
solving optimization problems on such graph classes have been studied in \cite{BuiXuanTV13,Vatshelle12}.

In this paper, we study two ways of using bounded LMIM-width in enumeration algorithms.  In Section~\ref{sec:lmim} we study the enumeration problem corresponding to an extended and colored version of
the well-known $(\sigma,\rho)$-domination problem, asking to enumerate all 1-minimal and all 1-maximal \Red $(\sigma,\rho)$-dominating sets.
This includes the enumeration of all minimal (total) dominating sets on graphs of bounded LMIM-width. We establish as our main result an enumeration algorithm with linear delay and polynomial space
for this problem. Our algorithm uses the enumeration (and counting) of paths in directed acyclic graphs.  In Section~\ref{sec:unit-square} we study the enumeration of all minimal dominating sets in
unit square graphs.  We first show that such graphs have bounded local LMIM-width.  An hereditary graph class $\mathcal{G}$ has bounded local LMIM-width if there is a function $f:\bN\to \bN$ such that
the LMIM-width of every graph $G$ in $\mathcal{G}$ is bounded by $f(diam)$ where $diam$ is the diameter of $G$. The notion of bounded local width has been studied for several width notions and in
particular in the area of Bidimensionality \cite{DemaineFHT04,DemaineHT06}. Then we show how to adapt the so-called flipping method developed by Golovach et al.~\cite{GolovachHKV14} to enumerate all
minimal dominating sets of a unit square graph in incremental polynomial time.

\section{Definitions and preliminaries}\label{sec:defs}
{\bf Graphs.}
The power set of a set $V$ is denoted by $2^V$. For two sets $A$ and $B$ we let $A\setminus B$ be the set $\{x\in A\mid x\notin B\}$, and if $X$ is a subset of a ground set $V$, we let $\comp{X}$ be
the set $V\setminus X$. We often write $x$ to denote the singleton set $\{x\}$. We denote by $\bN$ the set of positive or null integers, and let $\bN^*$ be $\bN\setminus\{0\}$.

A graph $G$ is a pair $(V(G),E(G))$ with $V(G)$ its set of vertices and $E(G)$ its set of edges. An edge between two vertices $x$ and $y$ is denoted by $xy$ (respectively $yx$). The subgraph of $G$
induced by a subset $X$ of its vertex set is denoted by $G[X]$, and we write $G\setminus X$ to denote the induced subgraph $G[V(G)\setminus X]$.  For $F\subseteq E(G)$, we denote by $G-F$ the subgraph
$(V(G),E(G)\setminus F)$.  The set of vertices that is adjacent to $x$ is denoted by $N_G(x)$, and we let $N_G[x]$ be the set $N_G(x)\cup \{x\}$.
For $U\subseteq V(G)$, $N_G[U]=\cup_{v\in U}N_G[v]$ and $N_G(U)=N_G[U]\setminus U$.
For a vertex $x$ and a positive integer $r$, $N_G^r[x]$ denotes the set of vertices at distance at most $r$ from $x$. Clearly, $N_G^1[x]=N_G[x]$. 
For two disjoint subsets $A$ and $B$ of $V(G)$, let $G[A,B]$ denote the graph with vertex set $A\cup B$ and the edge set  $\{uv\in E(G)\mid u\in A, v\in B\}$.
 Clearly, $G[A,B]$ is a bipartite graph and $\{A,B\}$ is its bipartition.  Recall that a set of edges $M$ is an \emph{induced matching} if end-vertices of distinct edges of $M$ are different and not adjacent. We denote by $\mim_G(A,B)$ the size of a maximum induced matching in $G[A,B]$.

 Let $G$ be a graph, and let $\Red,\Blue\subseteq V(G)$ such that $\Red\cup\Blue=V(G)$.  We refer to the vertices of $\Red$ as the \emph{red} vertices, the vertices of $\Blue$ as the \emph{blue}
 vertices, and we say that $G$ together with given sets $\Red$ and $\Blue$ is a \emph{colored graph}.  For simplicity, whenever we say that $G$ is a colored graph, it is assumed that the sets $\Red$
 and $\Blue$ are given.  Notice that $\Red$ and $\Blue$ not necessarily disjoint. In particular, it can happen that $\Red=\Blue=V(G)$; a non-colored graph $G$ can be seen as a colored graph with
 $\Red=\Blue=V(G)$. We deal with colored graphs because our algorithm for unit-square graphs requires as a subroutine an algorithm that takes as input a colored graph and enumerates all minimal
 subsets of red vertices that dominate the blue vertices.

A graph $G$ is an \emph{(axis-parallel) unit square} graph  if it is an intersection graph of squares in the plane with their sides parallel to the coordinate axis. These graphs also are known as the graphs of cubicity 2. We use the following equivalent definition, see e.g.~\cite{ChandranFS09}, 
in which each vertex $v$ of $G$ is represented by a point in $\mathbb{R}^2$. 
A graph $G$ is a unit square graph if there is a function $f\colon V(G)\rightarrow\mathbb{R}^2$ such that two vertices $u,v\in V(G)$ are adjacent in $G$ if and only if $\| f(u)-f(v)\|_{\infty}< 1$,
where the norm $\| \|_{\infty}$ is the $L_{\infty}$ norm. 
For a vertex $v\in V(G)$, we let $x_f(v)$ and $y_f(v)$ denote the $x$ and $y$-coordinate of $f(v)$ respectively.
We say that the point $(x_f(v),y_f(v))$ \emph{represents} $v$.
The function $f$ is called a \emph{realization} of the unit square graph. It is straightforward to see that for any unit square graph $G$, there is its realization $f\colon V(G)\rightarrow\mathbb{Q}^2$. 
We always assume that a unit square graph is given with its realization. 
It is NP-hard to recognize unit square graphs~\cite{Brie96}.
We refer to the survey of Brandst{\"a}dt, Le and Spinrad~\cite{BrandstdtLS99} for the definitions of all other graph classes mentioned in our paper.

\medskip
\noindent
{\bf Enumeration.}
Let ${\mathcal D}$ be a family of subsets of the vertex set of a given graph $G$ on $n$ vertices and $m$ edges. An \emph{enumeration algorithm} for ${\mathcal D}$ lists the elements of $\mathcal{D}$ without
repetitions.  The running time of an enumeration algorithm $\mathcal{A}$ is said to be \emph{output polynomial} if there is a polynomial $p(x,y)$ such that all the elements of $\mathcal{D}$ are listed
in time bounded by $p((n+m),|\mathcal{D}|)$. Assume now that $D_1,\ldots,D_\ell$ are the elements of $\mathcal{D}$ enumerated in the order in which they are generated by $\mathcal{A}$.  Let us denote
by $T({\mathcal A}, i)$ the time ${\mathcal A}$ requires until it outputs $D_i$, also $T({\mathcal A},\ell+1)$ is the time required by ${\mathcal A}$ until it stops. Let $delay({\mathcal A},1)=T({\mathcal A},1)$ and
$delay({\mathcal A},i)=T({\mathcal A},i)-T({\mathcal A},i-1)$. The {\it delay} of $\mathcal{A}$ is $\max\{delay({\mathcal A},i)\}$. Algorithm $\mathcal{A}$ runs in \emph{incremental polynomial} time if there is a
polynomial $p(x,i)$ such that $delay({\mathcal A},i)\leq p(n+m,i)$. Furthermore $\mathcal{A}$ is a \emph{polynomial delay} algorithm if there is a polynomial $p(x)$ such that the delay of ${\mathcal A}$ is at
most $p(n+m)$. Finally $\mathcal{A}$ is a \emph{linear delay} algorithm if $delay({\mathcal A},1)$ is bounded by a polynomial in $n+m$ and $delay({\mathcal A},i)$ is bounded by a linear function in $n+m$.

\medskip
\noindent
{\bf Linear maximum induced matching width.}
The notion of the \emph{maximum induced matching width} was introduced by Vatshelle~\cite{Vatshelle12} (see also~\cite{BelmonteV13}).  We will give the definition in terms of colored graphs and
restrict ourselves to the case of linear maximum induced matching width. Let $G$ be a colored $n$-vertex graph with $n\geq 2$ and let $x_1,\ldots,x_n$ be a linear ordering of its vertex set. For each
$1\leq i \leq n$, we let $A_i=\{x_1,x_2,\ldots x_i\}$ and $\comp{A}_i=\{x_{i+1},x_{i+2},\ldots x_n\}$. The \emph{maximum induced matching width} (\emph{MIM-width} for short) of $x_1,\ldots,x_n$ is
{\small
\begin{align*} 
\max\limits_{1\leq i \leq n-1} \max\{\mim_G(A_i\cap \Red,\comp{A}_i\cap\Blue),\mim_G(A_i\cap \Blue,\comp{A}_i\cap\Red)\}.
\end{align*}}

Notice that if $\Red=\Blue=V(G)$, \ie, if $G$ is an uncolored graph, the MIM-width of $x_1,\ldots,x_n$ is $\max\{\mim_G(A_i,\overline{A}_i)| i \in [n-1]\}$.  Consequently, The \emph{linear maximum
  induced matching width} (\emph{LMIM-width}) of $G$, denoted by $\lmimw(G)$, is the minimum value of the MIM-width taken over all linear orderings of $G$.

Belmonte and Vatshelle~\cite{BelmonteV13} proved that several important graph classes have bounded linear maximum induced matching width.

\begin{thm}\label{thm:bound-lmimw}
  For each of the following graph classes: interval graphs, permutation graphs, circular-arc graphs, circular permutation graphs, trapezoid graphs, convex graphs, and for fixed $k$, $k$-polygon
  graphs, Dilworth-$k$ graphs and complements of $k$-degenerate graphs, there is a constant $c$ such that $\lmimw(G)\leq c$ for any graph $G$ from the class. Moreover, the corresponding linear
  ordering of the vertices of MIM-width at most $c$ can be found in polynomial time.
\end{thm}

For example, the LMIM-width of an interval graph is 1 and the LMIM-width of a circular-arc and of a permutation graph is at most 2. Before continuing, let us show that if a graph class $\cG$ has
LMIM-width $c$, then the graph class $\cG\rq{}$ obtained from the graphs in $\cG$ by partitioning their vertex set into \Blue and \Red also has LMIM-width $c$. 

\begin{prop}\label{prop:bound}
If a colored graph $G\rq{}$ is obtained from a non-colored graph $G$, then 
$\lmimw(G\rq{})\leq \lmimw(G)$. 
\end{prop}

\begin{proof}
  Let $A\subset V(G)$. For the colored graph $G'$, $\max\{\mim_{G'}(A,\overline{A})\leq \mim_{G'}(A\cap \Red,\overline{A}\cap \Blue),\mim_{G'}(A\cap \Blue,\overline{A}\cap \Red)\}$. Because
  $\mim_{G'}(A\cap \Red,\overline{A}\cap \Blue)=\mim_{G}(A\cap \Red,\overline{A}\cap \Blue)\leq \mim_{G}(A,\overline{A})$ and
  $\mim_{G'}(A\cap \Blue,\overline{A}\cap \Red)=\mim_{G}(A\cap \Blue,\overline{A}\cap \Red)\leq \mim_{G}(A,\overline{A})$, $\mim_{G'}(A,\overline{A})\leq \mim_G(A,\overline{A})$ and the claim
  follows.
\end{proof}

We say that a graph class $\mathcal{G}$ has \emph{locally bounded LMIM-width} if there is a function $f:\bN\rightarrow \bN$ such that for any $G\in\mathcal{G}$ and every $u\in V(G)$,
$\lmimw(G[N_G^r[u]])\leq f(r)$ for all $r\in \bN$.

\medskip
\noindent
$\mathbf{(\sigma,\rho)}${\bf -domination.}
The \emph{$(\sigma,\rho)$-dominating set} notion was introduced by Telle and Proskurowski \cite{TelleP97} as a
generalization of dominating sets.  Indeed, many NP-hard domination type problems such as the problems $d$-Dominating Set, Independent Dominating Set and Total Dominating Set are special cases of the
$(\sigma,\rho)$-Dominating Set Problem.  See \cite[Table 1]{BuiXuanTV13} for more examples.
For technical reasons, we introduce  \Red $(\sigma,\rho)$-domination.
Let $\sigma$ and $\rho$ be subsets of $\bN$. Throughout this paper it is assumed that $\sigma$ and $\rho$ are finite or co-finite. Notice that it can happen that $\sigma$ is finite and $\rho$ is
co-finite and vice versa. We say that a subset $D$ of $V(G)$  \emph{$(\sigma, \rho)$-dominates} a vertex $u$ if 

\begin{align*}
  |N_G(u)\cap D| & \in \begin{cases} \sigma & \textrm{if $u\in D$},\\ \rho & \textrm{if $u\notin D$}. \end{cases}
\end{align*}

It is said that $D$ \emph{$(\sigma, \rho)$-dominates $U\subseteq V(G)$} if $D$  $(\sigma, \rho)$-dominates every vertex of $U$.

Let $G$ be a colored graph. A set of vertices $D\subseteq \Red$ is a \emph{\Red $(\sigma,\rho)$-dominating set} if $D$ $(\sigma,\rho)$-dominates $\Blue$.
If $\Red=\Blue=V(G)$, then a \Red $(\sigma,\rho)$-dominating set is a \emph{$(\sigma,\rho)$-dominating set}.

Notice that if $\sigma=\bN$ and $\rho=\bN^*$, then a set $D\subseteq V(G)$ $(\sigma,\rho)$-dominates a vertex $u$ if $u\in D$ or $u$ is adjacent to a vertex of $D$, \ie, the notion of $(\sigma,\rho)$-domination coincides with the classical domination in this case. Whenever we consider this case, we simply write that a set $D$ dominates a vertex or set and $D$ is a (\Red) dominating set omitting $(\sigma,\rho)$.

A \Red $(\sigma,\rho)$-dominating set $D$ of a graph $G$ is said \emph{minimal} if for any proper subset $D'\subset D$, $D'$ is not a \Red $(\sigma,\rho)$-dominating set, and we say that $D$ is
$1$-\emph{minimal} if for each vertex $x$ in $D$, $D\setminus x$ is not a \Red $(\sigma,\rho)$-dominating set.  Respectively, a \Red $(\sigma,\rho)$-dominating set $D$ is \emph{maximal} if for any
$D'$ such that $D\subset D'\subseteq \Red$, $D'$ is not a \Red $(\sigma,\rho)$-dominating set and $D$ is $1$-\emph{maximal} if for each vertex $x$ in $\Red\setminus D$, $D\cup \{x\}$ is not a \Red
$(\sigma,\rho)$-dominating set.  Clearly, every minimal or maximal \Red $(\sigma,\rho)$-dominating set is 1-minimal or 1-maximal, respectively, but not the other way around because the converse is not
true for arbitrary $\sigma$ and $\rho$. Observe however that every (\Red) $1$-minimal (total) dominating set is also a (\Red) minimal (total) dominating set. In Section \ref{sec:lmim} we enumerate
only $1$-minimal and $1$-maximal \Red $(\sigma,\rho)$-dominating sets.

Because our aim is to enumerate $1$-minimal or $1$-maximal \Red $(\sigma,\rho)$-do\-mi\-na\-ting sets, we need some certificate that a considered set is $1$-minimal or $1$-maximal respectively.
Let $D$ be a \Red $(\sigma,\rho)$-dominating set of a colored graph $G$. For a vertex $u\in D$, we say that the vertex $v\in \Blue$ is its \emph{certifying vertex} (or a \emph{certificate}) if $v$ is
not $(\sigma,\rho)$-dominated by $D\setminus \{u\}$, \ie, 

\begin{align*}
  |N_G(v)\cap (D\setminus \{u\})|\notin & \begin{cases} \sigma & \textrm{if $v\in (D\setminus \{u\})\cap \Blue$},\\ \rho & \textrm{if $v\in \Blue\setminus (D\setminus \{u\})$}. \end{cases}
\end{align*}

Respectively, for  a vertex $u\in \Red\setminus D$, the vertex $v\in \Blue$ is its \emph{certifying vertex} (or a \emph{certificate}) if $v$ is not $(\sigma,\rho)$-dominated by $D\cup \{u\}$.

Notice that because $D$ is a \Red $(\sigma,\rho)$-dominating set, if $v$ is a certificate for $u$, then $v\in N_G[u]$. Observe also that a vertex can be a certificate for many vertices and it can be a
certificate for itself. Notice that in the case of the classical domination, certificates are usually called \emph{privates} because a vertex is always a certificate for exactly one vertex, including
itself.  It is straightforward to show the following.

\begin{lem}\label{lem:cert}
  A set $D\subseteq \Red$ is a $1$-minimal \Red $(\sigma,\rho)$-dominating set of a colored graph $G$ if and only if each vertex $u\in D$ has a certificate.  Furthermore, $D$ is a $1$-maximal \Red
  $(\sigma,\rho)$-dominating set of $G$ if and only if each vertex $u\in \Red \setminus D$ has a certificate.
\end{lem}

\section{Enumerations for graphs of bounded LMIM-width}\label{sec:lmim}
In this section we prove the following. 

\begin{thm}\label{thm:mainthm} Let $(\sigma,\rho)$ be a pair of finite or co-finite subsets of $\bN$ and let $c$ be a positive integer.  For a colored graph $G$ 
  given with a linear ordering of $V(G)$ of MIM-width at most $c$, one can count in time bounded by $O(n^c)$, and enumerate with linear delay, all $1$-minimal (or $1$-maximal) \Red
  $(\sigma,\rho)$-dominating sets of $G$.
\end{thm} 

As a corollary of Theorems~\ref{thm:bound-lmimw} and \ref{thm:mainthm}, and Proposition \ref{prop:bound} we have the following.

\begin{cor}\label{thm:cormain} Let $(\sigma,\rho)$ be a pair of finite or co-finite subsets of $\bN$. Then, for every colored graph $G$ in one of the following graph classes, we can count in polynomial
  time, and enumerate with linear delay all $1$-minimal 
  (or $1$-maximal) 
  \Red $(\sigma,\rho)$-dominating sets of $G$: interval graphs, permutation graphs, circular-arc graphs, circular permutation graphs,
  trapezoid graphs, convex graphs, and for fixed $k$, $k$-polygon graphs, Dilworth-$k$ graphs and complements of $k$-degenerate graphs.
\end{cor}

The following corollary improves some known results in the enumeration of minimal transversals of some geometric hypergraphs (see e.g. \cite{Rauf11}). 

\begin{cor}\label{thm:cormain2} For every hypergraph $\cH$ being an interval hypergraph
or a circular-arc hypergraph one can count in polynomial time, and enumerate with linear delay, all minimal transversals of $\cH$.
\end{cor}
\begin{proof} For any of the considered hypergraphs, its incidence graph is a subgraph of 
an interval or a circular-arc graph. If we color the vertices of the hypergraph in \Red and the hyperedges in \Blue, then
  $X$ is a minimal transversal in the hypergraph if and only if it is a \Red $(\bN,\bN^*)$-dominating set.
\end{proof}

The remaining part of the section is devoted to the proof of Theorem \ref{thm:mainthm}. 
In Section~\ref{subsec:technical} we  give some technical definitions and lemmas  that are important for the definition of the DAG whose maximal paths correspond to the desired sets. In Section
\ref{subsec:mindag-defn} we define the DAG whose maximal paths correspond to the $1$-minimal \Red $(\sigma,\rho)$-dominating sets and then show that it can be constructed in polynomial time. We also
recall how to count in polynomial time, and enumerate with linear delay the maximal paths of a DAG.  We then explain in Section \ref{subsec:max} how the construction of Section \ref{subsec:mindag-defn} can
be rewritten for $1$-maximal \Red $(\sigma,\rho)$-dominating sets.

\subsection{Technical Definitions}\label{subsec:technical}

First because $\sigma$ and $\rho$ can be infinite, we need a finite way to check if a vertex is $(\sigma,\rho)$-dominated.  Let $d(\bN)=0$. For every finite set $\mu\subseteq \bN$, let
$d(\mu):=1 + \max\{a \mid a\in \mu\}$, and for every co-finite set $\mu\subseteq \bN$, let $d(\mu) := 1 +\max\{a\mid a\in \bN\setminus \mu\}$.  For finite or co-finite subsets $\sigma$ and $\rho$ of
$\bN$, we let $d(\sigma,\rho):=\max(d(\sigma),d(\rho))$. Given a subset $D$ of \Red, we can check if $D$ is a \Red $(\sigma,\rho)$-dominating set by computing
$|D\cap N_G(x)|$ up to $d(\sigma,\rho)$ for each vertex $x$ in $\Blue$ \cite{BuiXuanTV13}.
We need the following properties of certificates.

\begin{lem}\label{lem:cert-assign}
  Let $D$ be a \Red $(\sigma,\rho)$-dominating set of colored graph $G$.  If $v$ is a certificate for $u\in D$, then $v=u$ or $v$ is a certificate for all vertices of $N_G(v)\cap D$.  If $v$ is a certificate for
  $u\in\overline{D}\cap \Red$, then $v=u$ or $v$ is a certificate for all vertices of $N_G(v)\cap \overline{D}\cap \Red$.
\end{lem}

\begin{proof}
  Let $v\ne u$ be a certificate for $u\in D$ and let $D':=D\setminus \{u\}$. If $v\in D'$ and $|N_G(v)\cap D'|\notin \sigma$, then for any $w\in N_G(v)\cap D$, $|N_G(v)\cap (D\setminus
  \{w\})|=|N_G(v)\cap D'|\notin\sigma$.  If $v\notin D'$ and $|N_G(v)\cap D'|\notin \rho$, then for any $w\in N_G(v)\cap D$, $|N_G(v)\cap (D\setminus \{w\})|=|N_G(v)\cap D'|\notin\rho$.  The second
  claim can be proved by similar arguments.
\end{proof}

We define $\sigma^*:=\sigma\setminus\rho$ and $\rho^*:=\rho\setminus\sigma$.  Let also $\sigma^-:=\{i\in\sigma\mid i-1\notin\sigma\}$, $\rho^-:=\{i\in\rho\mid i-1\notin\rho\}$,
$\sigma^+:=\{i\in\sigma\mid i+1\notin\sigma\}$ and $\rho^+:=\{i\in\rho\mid i+1\notin\rho\}$.  By the definitions, we have the following property.

 \begin{lem}\label{lem:set-bound}
   The sets $\sigma^*,\rho^*, \sigma^-,\rho^-,\sigma^+,\rho^+$ are finite or co-finite. Also, $d(\sigma^*,\rho^*)\leq d(\sigma,\rho)$, $d(\sigma^-,\rho^-)\leq d(\sigma,\rho)+1$ and
   $d(\sigma^+,\rho^+)\leq d(\sigma,\rho)+1$. 
\end{lem}

By the definition of certificates, we have the next easy lemma.

\begin{lem}\label{lem:cert-intersect}
  Let $D$ be a \Red $(\sigma,\rho)$-dominating set of a colored graph $G$ and let $u\in \Red$ and $v\in \Blue$ be distinct vertices of $G$.  If $u\in D$, then $v$ is a certificate for $u$ if and only if 
  \begin{align*}
    |N_G(v)\cap D| &\in \begin{cases} \sigma^- & \textrm{if $v\in D$}\\ \rho^- & \textrm{if $v\notin D$.} \end{cases}
  \end{align*}

  If $u\notin D$, then $v$ is a certificate for $u$ if and only if 
  \begin{align*} 
    |N_G(v)\cap D|&\in \begin{cases} \sigma^+ & \textrm{if $v\in D$}\\  \rho^+ & \textrm{if $v\notin D$.} \end{cases}
  \end{align*}

  A blue vertex $v\in D$ is a certificate for itself if and only if $|N_G(v)\cap D|\in\sigma^*$.  A red vertex $v\notin D$ is a certificate for itself if and only if it is blue and $|N_G(v)\cap D|\in\rho^*$.
\end{lem}

Let $d\in \bN$ and let $A$ be a subset of the vertex set of a colored graph $G$. Two red subsets $X$ and $Y$ of $A$ are \emph{d-neighbor equivalent \wrt $A$}, denoted by $X\equiv_A^d Y$, if for all
$x\in \comp{A}\cap \Blue$ we have
\begin{align*}
  \min(d,|X\cap N_G(x)|) = \min(d,|Y\cap N_G(x)|).
\end{align*}

It is not hard to check that $\equiv_A^d$ is an equivalence relation and et us denote by $nec(\equiv_A^d)$ the number of equivalence classes of $\equiv_A^d$.  Belmonte and Vatshelle~\cite{BelmonteV13}
proved the following lemma restated in our setting.

\begin{lem}[\cite{BelmonteV13}]\label{lem:mim-neighborhoods} Let $d\in \bN$  and let $A$ be a subset of the vertex set of a colored graph $G$ such that $\mim_G(A\cap \Red,\comp{A}\cap \Blue)\leq
  k$. Then $nec(\equiv_{A}^d)\leq n^{d\cdot k}$.
\end{lem}

The next lemma is used to bound the number of information we have to store at each node of the DAG, which will consequently imply, combined with Lemma \ref{lem:mim-neighborhoods}, that the size of the
DAG is polynomial in the size of $G$.

\begin{lem}[\cite{BelmonteV13}]\label{lem:cert-bound} Let $G$ be a colored graph and let $A$ be a subset of $V(G)$. Then, $\mim_G(A\cap \Red,\comp{A}\cap\Blue) \leq k$ if and only if for every blue
  subset $S$ of $\comp{A}$ there is $C\subseteq S$ such that $N(C)\cap (A\cap \Red)=N(S)\cap (A\cap \Red)$ and $|C|\leq
  k$. 
\end{lem}

\subsection{Constructing the DAG for $1$-minimal sets}\label{subsec:mindag-defn}

Throughout this section we let $(\sigma,\rho)$ be a fixed pair of finite or co-finite subsets of $\bN$ and we let $G$ be a fixed $n$-vertex colored graph with $n\geq 2$. Let also $x_1,\ldots,x_n$ be a fixed
linear ordering of the vertex set of $G$ such that the MIM-width of $x_1,\ldots,x_n$ is bounded by a constant $c$.  Furthermore, for all $i\in \{1,2,\ldots ,n\}$, we let $A_i=\{x_1,x_2,\ldots x_i\}$
and $\overline{A_i}=\{x_{i+1},x_{i+2},\ldots x_n\}$. We furthermore let $d=d(\sigma,\rho)$.

We will follow the same idea as in \cite{BuiXuanTV13} where a minimum (or a maximum) $(\sigma,\rho)$-dominating set is computed, and we need for that to recall some definitions and lemmas (restated in
our setting) proved in \cite{BuiXuanTV13}.  For every $i\in\{1,\ldots,n\}$ and every subset $X$ of $A_i\cap \Red$, we denote by $rep_{A_i}^d(X)$ the lexicographically smallest set
$R\subseteq A_i\cap \Red$ such that $|R|$ is minimised and $R\equiv_{A_i}^d X$. Notice that it can happen that $R=\emptyset$.

\begin{lem}[\cite{BuiXuanTV13}] \label{lem:rep} For every $i\in \{1,\ldots,n\}$, one can compute a list $LR_i$ containing all representatives \wrt $\equiv_{A_i}^d$ in time
  $O(nec(\equiv_{A_i}^d)\cdot \log(nec(\equiv_{A_i}^d))\cdot n^2)$.  One can also compute a data structure that given a set $X\subseteq A_i\cap \Red$ in time
  $O(\log(nec(\equiv_{A_i}^d))\cdot |X|\cdot n)$ allows us to find a pointer to $rep_{A_i}^d(X)$ in $LR_i$. Similar statements hold for the list $LR_{\comp{i}}$ containing all representatives \wrt
  $\equiv_{\comp{A_i}}^d$.
\end{lem}

We will define a DAG, denoted by $DAG(G)$, the maximal paths of which correspond exactly to the $1$-minimal \Red $(\sigma,\rho)$-dominating sets of $G$.

\medskip

For $1\leq j \leq n$ and $C\subseteq A_j\cap \Blue$ (or $C\subseteq \comp{A}_j\cap \Blue$) we denote by $\cSG_j(C)$ (or by $\cGG_j(C)$) the set $X$ obtained from $C$ if we we initially set $X=C$ and
recursively apply the following rule: let $x$ be the greatest (or smallest) vertex in $X$ such that $N(X\setminus \{x\})\cap (\comp{A}_j \cap \Red)= N(X)\cap (\comp{A}_j\cap \Red)$ (or
$N(X\setminus \{x\})\cap (A_j\cap \Red) =N(X)\cap (A_j\cap \Red)$) and set $X=X\setminus \{x\}$.  Notice that $\cSG_j(C)$ and $\cGG_j(C)$ are both uniquely determined, and both have sizes bounded by
$c$ from Lemma \ref{lem:cert-bound}.  Observe also that if $C\subseteq A_j\cap \Blue$ (or $C\subseteq \comp{A}_j\cap \Blue$), then $\cSG_\ell(C\cup \{x_\ell\}) = \cSG_\ell(\cSG_j(C)\cup \{x_\ell\})$ for all
$\ell >j$ (or $\cGG_\ell(C\cup \{x_\ell\})= \cGG_\ell(\cGG_j(C)\cup \{x_\ell\})$ for all $\ell \leq j$). The constructors $\cSG_j$ and $\cGG_j$ are used to canonically choose certificates in order to
avoid redundancies.

Let $1\leq j < n$ and let $(R_j,R_j',C_j,C_j')\in LR_j\times LR_{\bar{j}}\times 2^{A_j\cap \Blue} \times 2^{\comp{A}_j\cap \Blue}$ and
$(R_{j+1},R_{j+1}',C_{j+1},C_{j+1}')\in LR_{j+1}\times LR_{\comp{j+1}}\times 2^{A_{j+1}\cap \Blue}\times 2^{\comp{A}_{j+1}\cap \Blue}$.  There is an \emph{$\varepsilon$-arc-1} from
$(R_j,R_j',C_j,C_j')$ to $(R_{j+1},R_{j+1}',C_{j+1},C_{j+1}') $ if

\begin{enumerate}[leftmargin=0.9cm]
\item[(1.1)] $R_j\equiv_{A_{j+1}}^d R_{j+1}$ and $R_j'\equiv_{\comp{A}_j}^d R_{j+1}'$, and 
\item[(1.2)] if ($x_{j+1}\notin \Blue$ or ($x_{j+1}\in \Blue$ and $|N(x_{j+1})\cap (R_j\cup R_{j+1}')|\in \rho$ and $|N(x_{j+1})\cap (R_j\cup R_{j+1}')|\notin \rho^-$)) then ($C_{j+1}=\cSG_{j+1}(C_j)$
  and $C_j'=\cGG_j(C_{j+1}')$), otherwise we should have ($|N(x_{j+1})\cap (R_j\cup R_{j+1}')|\in \rho^-$) and 
  \begin{enumerate}
  \item[(1.2.a)] if $N(x_{j+1})\cap (\comp{A}_{j+1}\cap \Red)\ne \emptyset$, then $C_{j+1}=\cSG_{j+1}(C_j\cup \{x_{j+1}\})$, else $C_{j+1}=\cSG_{j+1}(C_j)$, and 
  \item[(1.2.b)] if $N(x_{j+1}) \cap (A_j \cap \Red)\ne \emptyset$, then $C_j'=\cGG_j(C_{j+1}'\cup \{x_{j+1}\})$, else $C_j'=\cGG_j(C_{j+1}')$.
\end{enumerate}
\end{enumerate}

There is an \emph{$\varepsilon$-arc-2} from $(R_j,R_j',C_j,C_j')$ to $(R_{j+1},R_{j+1}',C_{j+1},C_{j+1}')$ if
\begin{enumerate}[leftmargin=0.9cm]
\item[(2.1)] $R_{j+1}\equiv_{A_{j+1}}^d (R_j\cup \{x_{j+1}\})$, $R_j'\equiv_{\comp{A}_j}^d (R_{j+1}'\cup \{x_{j+1}\})$, $x_{j+1}\in \Red$, ($|N(x_{j+1})\cap (R_j\cup R_{j+1}')|\in \sigma$ if
  $x_{j+1}\in \Blue$), and 
\item[(2.2)] if ($x_{j+1}\notin \Blue$ or ($x_{j+1}\in \Blue$ and $|N(x_{j+1})\cap (R_j\cup R_{j+1}')|\notin \sigma^-$)), then
  ($C_{j+1}=\cSG_{j+1}(C_j)$ and $C_j'=\cGG_j(C_{j+1}')$), otherwise we should have ($|N(x_{j+1})\cap (R_j\cup R_{j+1}')|\in \sigma^-$) and
  \begin{enumerate}
  \item[(2.2.a)] if $N(x_{j+1})\cap (\comp{A}_{j+1}\cap \Red)\ne \emptyset$, then $C_{j+1}=\cSG_{j+1}(C_j\cup \{x_{j+1}\})$, else $C_{j+1}=\cSG_{j+1}(C_j)$, and
  \item[(2.2.b)] if $N(x_{j+1}) \cap (A_j \cap \Red)\ne \emptyset$, then $C_j'=\cGG_j(C_{j+1}'\cup \{x_{j+1}\})$, else $C_j'=\cGG_j(C_{j+1}')$, and 
\end{enumerate}
\item[(2.3)] either ($N(x_{j+1})\cap (C_j \cup C_{j+1}' )\ne \emptyset$) or ($(x_{j+1}\in \Blue$ and $|N(x_{j+1})\cap (R_j\cup R_{j+1}')|\in \sigma^*$).
\end{enumerate}

\paragraph{\bf The nodes of $DAG(G)$.}  $(R,R',C,C',i)\in LR_i\times LR_{\bar{i}}\times 2^{A_i\cap \Blue}\times 2^{\comp{A}_i\cap \Blue}\times [n]$ is a node of $DAG(G)$ whenever $x_i\in \Red$, 
$C=\cSG_i(C)$ and $C'=\cGG_i(C')$.  We call $i$ the \emph{index} of $(R,R',C,C',i)$. Finally  $s=(\emptyset,\emptyset,\emptyset,\emptyset,0)$ is the  \emph{source node} and $t=(\emptyset,\emptyset,\emptyset,\emptyset,n+1)$ is the \emph{terminal node} of $DAG(G)$.

\paragraph{\bf The arcs of $DAG(G)$.}  There is an arc from the node $(R_0,R_0',C_0,C_0',j)$ to the node $(R_p,R_p',C_p,C_p',j+p)$ with $1\leq j < j+p\leq n$ if there exist tuples
$(R_1,R_1',C_1,C_1')$, \ldots, $(R_{p-1},R_{p-1}',C_{p-1},C_{p-1}')$ such that (1)~ for each $1\leq i \leq p-1$,
$(R_{i},R_{i}',C_{i},C_{i}')\in LR_{j+i}\times LR_{\widebar{j+i}}\times 2^{A_{j+i}\cap \Blue}\times 2^{\comp{A}_{j+i}\cap \Blue}$ and there is an $\varepsilon$-arc-1 from
$(R_{i-1},R_{i-1}',C_{i-1},C_{i-1}')$ to $(R_{i},R_{i}',C_{i},C_{i}')$, and~ (2) there is an $\varepsilon$-arc-2 from $(R_{p-1},R_{p-1}',C_{p-1},C_{p-1}')$ to $(R_p,R_p',C_p,C_p')$.
\medskip
 
There is an arc from the source node to a node $(R,R',C,C',j)$ if
({\small $S=\{x\in (A_j\cap \Blue)\setminus \{x_j\}\mid N(x)\cap (\comp{A}_j\cap \Red)\ne \emptyset$ and $|N(x)\cap (\{x_j\}\cup R')|\in \rho^-\}$})
\begin{enumerate}[leftmargin=0.8cm]
\item[(S1)] $\{x_j\} \equiv_{A_j}^d R$ and $(\{x_j\} \cup R')$ $(\sigma,\rho)$-dominates $A_j\cap \Blue$,
\item[(S2)] if ($x_j\in \Blue$ and $|N(x_j)\cap R'|\in \sigma^-$) then $C=\cSG_j(S\cup \{x_j\})$, otherwise $C=\cSG_j(S)$, and
\item[(S3)] either ($N(x_j)\cap (C'\cup C)\ne \emptyset$) or ($x_j\in \Blue$ and $|N(x_j)\cap R'|\in \sigma^*$).
\end{enumerate}

\medskip 
There is an arc from a node  $(R,R',C,C',j)$ to the terminal node if 
\begin{enumerate}[leftmargin=0.8cm]
\item[(T1)] $|N(x)\cap R|\in \rho$ for each $x\in \comp{A}_{j+1}\cap \Blue$, and
\item[(T2)] $C'=\cGG_j(\{x\in \comp{A}_j\cap \Blue\mid N(x)\cap (A_j\cap \Red)\ne \emptyset\ \text{and}\ |N(x)\cap R|\in \rho^-\})$. 
\end{enumerate}

\begin{lem}\label{lem:size-dag} $DAG(G)$ is a DAG and can be constructed in time $O(n^{c\cdot d})$.\end{lem}

\begin{proof}
  An arc is always oriented from a node $(R,R',C,C',j)$ to 
  $(\hat{R},\hat{R'}, \hat{C},\hat{C'},j+p)$ with $p\geq1$. Therefore, we cannot create circuits, \ie, $DAG(G)$ is a DAG.
 
  For each index $1\leq i \leq n$ and each node $(R,R',C,C',i)$ of index $i$ we know by \cite[Lemma 1]{BelmonteV13} that $|C|,|C'|\leq c$. Hence, the number of nodes of $DAG(G)$ of index $i$ is
  $O(n^{c\cdot d})$ since $nec(\equiv_{A_i}^d)\leq n^{d\cdot c}$ by Lemma \ref{lem:mim-neighborhoods}.  Now, constructing the arcs from the source node can be done in $O(n^{c\cdot d})$ time since it
  suffices to check for each node $(R,R',C,C',j)$ if conditions (S1)-(S3) are satisfied, which can be done trivially in polynomial time with the help of Lemma \ref{lem:rep}. Similarly, since the
  conditions (T1) and (T2) can be checked in polynomial time, the incoming arcs to the terminal node can be constructed in $O(n^{c\cdot d})$ time.

  Now, to construct an arc from $(R_0,R'_0,C_0,C'_0,j)$ to $(R_p,R'_p, C_p,C'_p,j+p)$ we do as follows. For $0\leq i \leq p-1$ we let $\cF_i$ be a queue and we put
  $(R_0,R_0',C_0,C_0')$ in $\cF_0$. Now, for $1\leq i \leq p-1$, pull $(R_{i-1},R_{i-1}',C_{i-1},C_{i-1}')$ from $\cF_{i-1}$, and for each $(R,R',C,C')$ with $R\in LR_{j+i}$, $R'\in LR_{\bar{j+i}}$,
  $C\subseteq A_{j+i}\cap \Blue$, $C'\subseteq \comp{A}_{j+i}\cap \Blue$ with $|C|,|C'|\leq c$ such that there is an $\varepsilon$-arc-1 from $(R_{i-1},R_{i-1}',C_{i-1},C_{i-1}')$ to
  $(R,R',C,C')$, then put $(R,R',C,C')$ in $\cF_i$.  Now, by the definition of an arc in $DAG(G)$ there is an arc from $(R_0,R'_0,C_0,C'_0,j)$ to $(R_p,R'_p, C_p,C'_p,j+p)$ if and only if there is one\\
  $(R_{p-1},R_{p-1}',C_{p-1},C_{p-1}')$ in $\cF_{p-1}$ such that there is an $\varepsilon$-arc-2 from\\ $(R_{p-1},R_{p-1}',C_{p-1},C_{p-1}')$ to $(R_p,R'_p, C_p,C'_p,j+p)$.  Now, the size of each
  $\cF_i$ is bounded by $O(n^{c\cdot d})$, and since the conditions of $\varepsilon$-arc-1 and $\varepsilon$-arc-2 can be checked in $O(n^{c\cdot d})$ time with the help of Lemma \ref{lem:rep}, 
  we can check if there is an arc between two nodes in time $O(n^{c\cdot d})$.
\end{proof}

We now prove that there is a one-to-one correspondence between the maximal paths of $G$ and the $1$-minimal \Red $(\sigma,\rho)$-dominating sets of $G$. If $P=(s,v_1,v_2,\ldots,v_p,t)$ is a path in
$DAG(G)$, then the \emph{trace of $P$}, denoted by $\trace(P)$, is defined as $\{x_{j_1},x_{j_2},\ldots,x_{j_p}\}$ where for all $i\in \{1,2,\ldots ,p\}$, $j_i$ is the index of the node $v_i$. 

The following two lemmas are implied by the definition of the d-neighbor equivalence and Lemma~\ref{lem:set-bound}.

\begin{lem} \label{lem:equiv} Let $(\mu,\mu') \in \{(\sigma,\rho),(\sigma^*,\rho^*),(\sigma^-,\rho^-),(\sigma^+,\rho^+)\}$. Let also $i\in \{1,\ldots,n\}$, and $X\subseteq A_i\cap \Red$ and
  $Y,Y'\subseteq \overline{A_i}\cap \Red$.  If $Y'\equiv_{\overline{A}_{i}}^d Y$ then $X\cup Y$ $(\mu,\mu')$-dominates $A_i\cap \Blue$ if and only if $X\cup Y'$ $(\mu,\mu')$-dominates $A_i\cap
  \Blue$.
  Symmetrically, if $X,X'\subseteq A_i\cap \Red$ and $Y\subseteq \overline{A_i}\cap \Red$, and $X'\equiv_{A_{i}}^d X$, then $X\cup Y$ $(\mu,\mu')$-dominates $\comp{A}_i\cap \Blue$
  if and only if $X'\cup Y$ $(\mu,\mu')$-dominates $\comp{A}_i\cap \Blue$.
\end{lem}

\begin{lem}\label{lem:combine} Let $(\mu,\mu') \in \{(\sigma,\rho),(\sigma^*,\rho^*),(\sigma^-,\rho^-),(\sigma^+,\rho^+)\}$, $i\in \{1,\ldots,n\}$ and let $Z\subseteq \Red$. Let also
  $X\subseteq A_{i-1}\cap \Red$ and $Y\subseteq \comp{A_{i}}\cap \Red$. If $X \equiv_{A_{i-1}}^d (Z\cap A_{i-1})$ and $Y\equiv_{\overline{A}_{i}}^d (Z\cap \overline{A}_{i})$, then $Z$
  $(\mu,\mu')$-dominates $\{x_i\}$ if and only if $(X\cup Y\cup (Z\cap \{x_{i}\})$ $(\mu,\mu')$-dominates $\{x_i\}$.
\end{lem}

The next lemma shows that two maximal paths in $DAG(G)$ give rise to two different $1$-minimal \Red $(\sigma,\rho)$-dominating sets. 

\begin{lem}\label{lem:path-dom} If there is a path $P=(s,v_1,\ldots,v_k,t)$ in $DAG(G)$, then $\trace(P)$ is a $1$-minimal \Red $(\sigma,\rho)$-dominating set of
  $G$. Moreover, $\trace(P)\ne \trace(P')$ for any other path $P'=(s,v_1',\ldots,v_k',t)$ in $DAG(G)$.
\end{lem}

\begin{proof} Let $P=(s,v_1,\ldots,v_k,t)$ and $\trace(P)=\{x_{j_1},\ldots,x_{j_k}\}$. We will first prove by induction that for each $1\leq i \leq k$, the set
  $D_i=\{x_{j_1},\ldots,x_{j_i}\}$ satisfies the following properties (with $v_i=(R_{j_i},R_{j_i}',C_{j_i},C_{j_i}',j_i)$)
  \begin{itemize}
  \item[(i)] $R_{j_i}\in LR_{j_i}$, $R_{j_i}'\in LR_{\bar{j_i}}$;
  \item[(ii)]  $D_i\equiv_{A_{j_i}}^d R_{j_i}$, 
    \item[(iii)] $D_i\cup R_{j_i}'$ $(\sigma,\rho)$-dominates $A_{j_i}\cap \Blue$;
    \item[(iv)] Each $u\in D_i$ is either adjacent to a vertex from $C_{j_i}'$, or has a certificate in $A_{j_i}\cap \Blue$.
    \item[(v)]  $C_{j_i}=\cSG_{j_i}(S_i)$, where $S_i$ is the set of vertices in $A_{j_i}\cap \Blue$ that are certificates and have a neighbor in $\comp{A}_{j_i}\cap \Red$;
  \end{itemize}

  By the definition of a node in $DAG(G)$, the property (i) is true for all $1\leq i \leq k$. So, let us prove the properties (ii)-(v). By the definition of the arcs from the source node, we can
  easily check that the properties (ii)-(v) are all verified for $i=1$. So, let us assume now that they are true for all $i < \ell \leq k$ and let us prove it for $\ell$.

  If there is an arc from $v_{\ell-1}$ to $v_\ell$, then there should exist $(R_s,R_s',C_s,C_s')$ for $j_{\ell-1}+1\leq s \leq j_\ell-1$ such that there is an $\varepsilon$-arc-1 from
  $(R_{s-1},R_{s-1}',C_{s-1},C_{s-1}')$ to $(R_s,R_s',C_s,C_s')$ for each $j_{\ell-1}+1\leq s \leq j_\ell-1$, and there is an $\varepsilon$-arc-2 from
  $(R_{j_\ell-1},R_{j_\ell-1}',C_{j_\ell-1},C_{j_\ell-1}')$ to $(R_{j_\ell},R_{j_\ell}',C_{j_\ell},C_{j_\ell}')$. By the conditions (1.1) and (2.1) we can conclude that
  $D_{j_\ell}\equiv_{A_{j_\ell}}^d R_{j_\ell}$ because $D_{j_{\ell-1}}\equiv_{A_s}^d R_s$ for all $j_{\ell-1}+1\leq s \leq j_\ell-1$ by the condition (1.1) and
  $R_{j_\ell}\equiv_{A_{j_\ell}}^d (R_{j_\ell-1}\cup \{x_{j_\ell}\})$ by the condition (2.1).

  Because $R'_s\equiv_{\comp{A}_{s-1}}^d R'_{s-1}$ for each $j_{\ell-1}+1\leq s \leq j_\ell-1$ by (1.1) and $R_{j_\ell-1}'\equiv_{\comp{A}_{j_\ell-1}}^d R_{j_\ell}'\cup \{x_{j_\ell}\}$ by (2.1) we
  can conclude with inductive hypothesis, Lemmas \ref{lem:equiv} and \ref{lem:combine},  and the conditions (1.2) and (2.1) that for each
  $j_{\ell-1}+1\leq s \leq j_\ell$ whenever $x_s\in \Blue$ it is $(\sigma,\rho)$-dominated by $D_{j_\ell}\cup R_{j_\ell}'$, thus proving (iii). 

  In order to check (iv) and (v), we let $D_s=D_{j_{\ell-1}}$ for each $j_{\ell-1}+1\leq s \leq j_\ell-1$. Then, for each $j_{\ell-1}+1\leq s\leq j_\ell$, the following easy facts can be derived from
  Lemmas \ref{lem:equiv} and \ref{lem:combine}, the definition of the $d$-neighbor equivalence and the fact that $D_s \cap R'_s= \emptyset$.

  \begin{enumerate}
  \item For each $v\in \comp{A}_{s}\cap \Blue$, we have $(D_s\cup R_s')$ $(\sigma^-,\rho^-)$-dominates $\{v\}$ if and only if $(R_s\cup R_s')$ $(\sigma^-,\rho^-)$-dominates $\{v\}$; and
    $|N(v)\cap D_s|\ne 0$ if and only if $|N(v)\cap R_s|\ne 0$.
  \item For each $v\in A_{s-1}\cap \Blue$, we have $(D_s\cup R_s')$ $(\sigma^-,\rho^-)$-dominates $\{v\}$ if and only if $(D_{s-1}\cup R_{s-1}')$ $(\sigma^-,\rho^-)$-dominates $\{v\}$; 
  \item $(D_s\cup R_s')$ $(\sigma^-,\rho^-)$-dominates $\{x_s\}$ if and only if $(R_{s-1}\cup R_s')$ $(\sigma^-,\rho^-)$-dominates $\{x_s\}$. 
  \item Each $u\in D_{s-1}$ either has a certificate in $A_{s}$ \wrt $D_s$ or is adjacent to a vertex from $C_s'$. Indeed, either it has by induction a certificate $v$ from $A_{s-1}$ \wrt $D_{s-1}$
    and by (2) and Lemma \ref{lem:cert-intersect} the vertex $v$ is still a certificate for $u$ \wrt $D_s$, or $u$ is adjacent to some vertex in $C_{s-1}'$ and then by induction and the conditions
    (1.2) and (2.2) either it is adjacent to some vertex in $C_s'$ or it is adjacent to $x_s$ which is $(\sigma^-,\rho^-)$-dominated by $D_s\cup R_s'$ following fact 3. 
  \end{enumerate}

  From the facts 1. and 2. we can conclude that $(D_s\cup R'_s)$ $(\sigma^-,\rho^-)$-dominates $C_s$ for all $j_{\ell-1}+1\leq s \leq j_\ell$. Hence, $(D_{j_\ell}\cup R'_{j_\ell})$
  $(\sigma^-,\rho^-)$-dominates $C_{j_\ell}$. Moreover, from the fact 4. we know that each $u\in D_{j_{\ell-1}}$ is either adjacent to a vertex in $C_{j_\ell}'$ or has a certificate \wrt $D_{j_\ell}$ in
  $A_{j_{\ell}}\cap \Blue$. In order to prove that (iv) is satisfied it remains then to check that $x_{j_\ell}$ has a certificate in $A_{j_\ell}\cap \Blue$ or has a neighbor in $C_\ell'$. But this is
  guaranteed with the existence of the arc $v_{\ell-1}$ to $v_\ell$ by the conditions (1.2.a), (2.2.a), (2.3),  and the properties (iv)-(v) by inductive hypothesis.

  In order to check the condition (v) it is sufficient to notice that whenever $x_s$ is $(\sigma^-,\rho^-)$-dominated by $(D_s\cup R_s')$ for each $j_{\ell-1}+1\leq s \leq j_\ell$, by the condition
  (1.2.a) and (2.2.a) $C_s=\cSG_s(C_{s-1}\cup x_s)$, and this guarantees by inductive hypothesis, the fact 2. and Lemma \ref{lem:cert-intersect} that $C_s$ is exactly $\cSG_s(S_s)$ where $S_s$ is the
  set of vertices in $A_s\cap \Blue$ that are certificates and have a neighbor in $\comp{A}_s$. 

  To end the proof we need to prove that whenever $\trace(P)= \trace(P')$ for any other path $P'$ from the source node to the terminal node, then $P=P'$.  For that we prove by induction that
  $C_{j_i}'=\cGG_{j_i}(S_i')$ where $S_i'$ is the set of vertices in $\comp{A}_{j_i}\cap \Blue$ that are $(\sigma^-,\rho^-)$-dominated by $D_{j_i}\cup R_{j_i}'$ and have a neighbor in
  $A_{j_i}\cap \Red$. By the condition (T2) in the definition of an arc to the terminal node this is satisfied by $C_{j_k}'$. So, if we assume that $C_{j_i}'=\cGG_{j_i}(S_i')$ for all $\ell < i \leq
  k$, then as for the condition (v) the inductive hypothesis, the fact 2. and Lemma \ref{lem:cert-intersect} guarantees that $C_{j_\ell}'$ is exactly $\cGG_{j_\ell}(S_\ell')$.

  So now for each $i$ the sets $C_{j_i}$ and $C_{j_i}'$ are uniquely determined by $\trace(P)$, which means that whenever $\trace(P)=\trace(P')$ because $\equiv_{A}^d$ is an equivalence relation we
  should conclude that $P=P'$. 
\end{proof}

The following lemma tells that to every $1$-minimal \Red $(\sigma,\rho)$-dominating set corresponds a maximal path in $DAG(G)$. 

\begin{lem}\label{lem:dom-path} If $G$ has a $1$-minimal \Red $(\sigma,\rho)$-dominating set $D$, then there is a path $P=(s,v_1,v_2,\ldots,v_k,t)$ in $DAG(G)$ such that $D=\trace(P)$.
\end{lem}

\begin{proof} Let $D=\{x_{j_1},\ldots,x_{j_k}\}$ such that $j_1<j_2<\cdots < j_k$. For each $i\in \{1,\ldots,k\}$, we let $D_i=\{x_{j_1},\ldots,x_{j_i}\}$.  Let also $R_{j_i}\in LR_{j_i}$ be such
  that $R_{j_i}\equiv_{A_{j_i}}^d D_i$, and let $R_{j_i}'\in LR_{\bar{j_i}}$ be such that $ R_{j_i}' \equiv_{\comp{A}_{j_i}}^d (D\cap\comp{A}_{j_i})$.  For each $i$ we let
  $C_{j_i}=\cSG_{j_i}(S_i)$ where $S_i$ is the set of vertices in $A_{j_i}\cap \Blue$ that are certificates and have a neighbor in $\comp{A}_{j_i}\cap \Red$ and similarly let
  $C_{j_i}'=\cGG_{j_i}(S_i')$ where $S_i'$ is the set of vertices in $\comp{A}_{j_i}\cap \Blue$ that are certificates and have a neighbor in $A_{j_i}\cap \Red$. Hence,
  $v_i=(R_{j_i},R_{j_i}',C_{j_i},C_{j_i}',j_i)$ is a node of $DAG(G)$ for each $1\leq i \leq k$.

  We first observe that there is an arc from the source node $s$ to $v_1$. Indeed, by the definition of $R_{j_i},R_{j_i}'$, we have that $\{x_{j_1}\}\equiv_{A_{j_1}}^d R_{j_1}$ and
  $\{x_{j_1}\}\cup R_{j_1}'$ $(\sigma,\rho)$-dominates $A_{j_1}$, and by the choices of $C_{j_i}$ and $C_{j_i}'$ the condition (S2) is satisfied and since
  $x_{j_1}$ has a certificate \wrt $D$, the condition (S3) is satisfied. For similar reasons one can prove that there is an arc from $v_k$ to the terminal node $t$. 

  We now claim that there is an arc from $v_i$ to $v_{i+1}$ for $1\leq i < k$. For each $j_{i}< s < j_{i+1}$, we let $(R_s,R_s',C_s,C_s')$ be such that $R_s\equiv_{A_{s}}^d R_{s-1}$,
  $R_s'\equiv_{\comp{A}_{s-1}}^d R_{s-1}'$, and $R_{j_{i+1}}\equiv_{A_{j_{i+1}}}^d R_{j_{i+1}-1}\cup \{x_{j_{i+1}}\}$ and
  $R_{j_{i+1}-1}'\equiv_{\comp{A}_{j_{i+1}-1}}^d R_{j_{i+1}}'\cup \{x_{j_{i+1}}\}$. It is straightforward to prove by induction on $j_{i+1}-j_i$ that  there exists an $\varepsilon$-arc-1 from
  $(R_{s-1},R_{s-1}',C_{s-1},C_{s-1}')$ to $(R_s,R_s',C_s,C_s')$ for each $j_i<s<j_{i+1}$ and there is an $\varepsilon$-arc-2 from $(R_{j_{i+1}-1}, R_{j_{i+1}-1}',C_{j_{i+1}-1},C_{j_{i+1}-1}')$ to
  $(R_{j_{i+1}}, R_{j_{i+1}}',C_{j_{i+1}},C_{j_{i+1}}')$. 
\end{proof}

By Lemmas \ref{lem:path-dom} and \ref{lem:dom-path} we can state the following. 

\begin{prop}\label{prop:path-dom}
Let $\cP$ be the set of paths in $DAG(G)$ from the source node to the terminal node. The mapping which associates with every $P\in \cP$ $\trace(P)$ is a one-to-one correspondence with the
set of $1$-minimal \Red $(\sigma,\rho)$-dominating sets.
\end{prop}

By Proposition \ref{prop:path-dom} it suffices to count and enumerate the traces of the maximal paths in $DAG(G)$. We will now explain how to count and then use the counting to enumerate the traces of
these paths in $DAG(G)$. We start from a topological ordering of $DAG(G)$, say $s=v_1,v_2,\ldots,v_m=t$. Since $DAG(G)$ is a DAG, any arc is of the form $(v_i,v_j)$ with $i<j$. The counting will
follow this topological ordering.  We initially set $Np(v)=-1$ for all nodes $v\ne t$ and we set $Np(v_m)=1$. For each $j<m$ we let
\begin{align*}
  Np(v_j) &= \sum\limits_{\substack{(v_j,v_\ell)\in E(DAG(G))\\ Np(v_\ell)\ne -1}} Np(v_\ell).
\end{align*}

\begin{fact}\label{fact:compute} One can compute the values of $Np(v_j)$
  for all $j\in \{1,2,\ldots ,m\}$ in time $O(n^{c\cdot d})$.
\end{fact}

\begin{proof} By induction on $j$. By definition $Np(v)$ can be computed in time $O(1)$ for all $v$ that have exactly one outgoing arc, which
  enters $t$. For every $j<m$, in order to compute $Np(v_j)$ we first set a counter $\mathbf{nb}$ to $0$, and add $Np(v_\ell)$ to $\mathbf{nb}$ whenever $(v_j,v_\ell)\in E(DAG(G))$ and
  $Np(v_\ell)\ne -1$. We finally set $Np(v_j)$ to $\mathbf{nb}$. The correctness of the computation of $Np(v_j)$ follows from the definition. Since the degree of a node is bounded by $O(n^{c\cdot d})$,
  we can update $\mathbf{nb}$ in time $O(n^{c\cdot d})$.  Now since the number of nodes and of arcs is bounded by $O(n^{c\cdot d})$, we obtain the claimed running time. 
\end{proof}

For $1\leq j \leq m$, we let $\cS_j=\{P\mid P$ is a path starting at $v_j$ and ending at $t\}$.   One can prove easily by induction the following. 

\begin{lem}\label{lem:disjoint} $\cS_j = \biguplus\limits_{\substack{(v_j,v_\ell)\in E(DAG(G))\\ Np(v_\ell)\ne -1}} \{v_j+P\mid P\in \cS_\ell\}$ for each $1\leq j \leq m$.
\end{lem}

The following follows directly from the definition of $Np(v_j)$ and Lemma \ref{lem:disjoint}.

\begin{lem}\label{lem:count-paths} $|\cS_j|= Np(v_j)$ for each $1\leq j \leq m$. 
\end{lem}

\begin{thm}\label{thm:counting}  One can count the number of $1$-minimal \Red $(\sigma,\rho)$-dominating sets of a given graph $G$  in time $O(n^{c\cdot d})$.
\end{thm}

\begin{proof} We first construct the DAG $DAG(G)$ and by Lemma \ref{lem:size-dag} this can be done in time $O(n^{c\cdot d})$.  By Proposition \ref{prop:path-dom} the mapping which
  associates with every path $P\in \cS_1$ its trace $\trace(P)$ is a one-to-one correspondence between $\cS_1$ and all the $1$-minimal \Red $(\sigma,\rho)$-dominating set in $G$. So, it is enough to
  determine the size of $\cS_1$.  By Lemma \ref{lem:count-paths}, $|\cS_1|= Np(s)$ and since by Fact \ref{fact:compute} we can compute in time $O(n^{c\cdot d})$ all the values $Np(v_j)$ for all
  $1\leq j \leq m$, we conclude that one can compute in time $O(n^{c\cdot d})$ the number of $1$-minimal \Red $(\sigma,\rho)$-dominating sets in $G$.
\end{proof}

We now turn to the enumeration of the $1$-minimal \Red $(\sigma,\rho)$-dominating sets. For each node $v$ in $DAG(G)$ of index $j$ we denote by $\vertex(v)$ the vertex $x_j$ of $G$. The algorithm is
depicted in Figures \ref{fig:alg-enum} and \ref{fig:alg-enum-path}. The algorithm consists in enumerating the paths in $\cS_1$ in a Depth-First Search manner.

\begin{figure}[!h]
\begin{framed}
\begin{small}
\begin{tabbing}
{\bf Algorithm} {\sf EnumMinDom}$(DAG(G))$\\
1. Remove all nodes $v$ such that $Np(v)=-1$\\
2. {\bf for} each $(s,v_i)\in E(DAG(G))$ \\
3. \ \ \ {\sf EnumPath}$(DAG(G),\{\vertex(v_i)\}, v_i)$\\
4. {\bf end for}
\end{tabbing}
\end{small}
\end{framed}
\caption{The enumeration of $1$-minimal \Red $(\sigma,\rho)$-dominating sets}
\label{fig:alg-enum}
\end{figure}

\begin{figure}[!h]
\begin{framed}
\begin{small}
\begin{tabbing}
{\bf Algorithm} {\sf EnumPath}$(DAG(G), S\subseteq V(G),v_i)$\\
1. {\bf if} $v_i=t$, then {\bf output $S$ and stop}\\
2. {\bf for} each $(v_i,v_j)\in E(DAG(G))$ \\
3. \ \ \ {\sf EnumPath}$(DAG(G), S\cup \{\vertex(v_j)\},v_j)$\\
4. {\bf end for}
\end{tabbing}
\end{small}
\end{framed}
\caption{The enumeration of $\cS_i$}
\label{fig:alg-enum-path}
\end{figure}

\begin{thm}\label{thm:enum} We can enumerate all the $1$-minimal \Red $(\sigma,\rho)$-dominating sets of a given graph $G$ with linear delay and with polynomial space. 
\end{thm}

\begin{proof} 
  First notice that after removing all the nodes $v$ such that $Np(v)=-1$, every remaining node is in a path from the source node to the terminal node. Now, it is easy to prove by induction using Lemmas
  \ref{lem:disjoint} and \ref{lem:count-paths} that the algorithm \textsf{EnumPath}$(DAG(G), S,v_i)$ uses $O(n^{c\cdot d})$ space and enumerates the set $\{S\cup P\mid P\in \cS_i\}$, the delay between
  two consecutive outputs $P_1$ and $P_2$ bounded by $O(|P_2\setminus P_1|)$. In fact if before calling \textsf{EnumPath} we order the out-neighbors of each node following their distances to the
  terminal node and uses this ordering in the recursive calls we guarantee that the time between the output of $P$ and the next output $Q$ is bounded by $O(|Q|)$. Therefore, the algorithm
  \textsf{EnumMinDom}$(DAG(G))$ enumerates, with same delay as \textsf{EnumPath} the set of $1$-minimal \Red $(\sigma,\rho)$-dominating sets and uses $O(n^{c\cdot d})$ space. 
\end{proof}

\subsection{Maximal sets}\label{subsec:max}
We now explain how to construct the DAG $DAGM(G)$ so that the maximal paths from the source node to the terminal node corresponds to the $1$-maximal \Red $(\sigma,\rho)$-dominating sets, and conversely each
$1$-maximal \Red $(\sigma,\rho)$-dominating set corresponds to such a path. 
The difference with the case of $1$-minimal $(\sigma,\rho)$-dominating set is that now we have to ensure that $x_j$ has a certificate when $x_j$ is not included in a partial solution.

\bigskip

\paragraph{\bf The nodes of $DAGM(G)$.}  $(R,R',C,C',i)\in LR_i\times LR_{\bar{i}}\times 2^{A_i\cap \Blue}\times 2^{\comp{A}_i\cap \Blue}\times [n]$ is a node of $DAGM(G)$ whenever $x_i\in \Red$, 
$C=\cSG_i(C)$ and $C'=\cGG_i(C')$.  We call $i$ the \emph{index} of $(R,R',C,C',i)$. Finally  $s=(\emptyset,\emptyset,\emptyset,\emptyset,0)$ is the  \emph{source node} and
$t=(\emptyset,\emptyset,\emptyset,\emptyset,n+1)$ is the \emph{terminal node} of $DAG(G)$.

\bigskip

\paragraph{\bf The arcs of $DAG(G)$} There is an arc from $(R_0,R'_0,C_0,C'_0,j)$ to $(R_P,R_p',C_p,C_p',j+p)$ with $j<j+p\leq n$ if there exist $(R_1,R'_1,C_1,C'_1)$, \ldots,
$(R_{p-1},R_{p-1}',C_{p-1},C_{p-1}')$ such that $(R_i,R_i',C_i,C_i')\in LR_{j+i}\times LR_{\bar{j+i}}\times 2^{A_{j+i}\cap \Blue}\times 2^{\comp{A}_{j+i}\cap \Blue}$ for all $1\leq i \leq p-1$, and

\begin{enumerate}[leftmargin=0.8cm]
\item[(A1)] for each $0\leq i \leq p-2$,
  \begin{enumerate}
  \item[(1.1)] $R_{i}\equiv_{A_{j+i+1}}^d R_{i+1}$ and $R_{i}'\equiv_{\comp{A}_{j+i=1}}^d R_{i+1}'$, and 
  \item[(1.2)] if ($x_{j+i+1}\notin \Blue$ or ($x_{j+i+1}\in \Blue$ and $|N(x_{j+i+1})\cap (R_i\cup R_{i+1}')|\in \rho$ and $|N(x_{j+i+1})\cap (R_i\cup R_{i+1}')|\notin \rho^+$)) then ($C_{i+1}=\cSG_{j+i+1}(C_i)$
    and $C_i'=\cGG_{j+i}(C_{i+1}')$), otherwise we should have ($|N(x_{j+i+1})\cap (R_i\cup R_{i+1}')|\in \rho^+$) and 
    \begin{enumerate}
    \item[(1.2.a)] if $N(x_{j+i+1})\cap (\comp{A}_{j+i+1}\cap \Red)\ne \emptyset$, then $C_{i+1}=\cSG_{j+i+1}(C_i\cup \{x_{j+i+1}\})$, else $C_{i+1}=\cSG_{j+i+1}(C_i)$, and 
    \item[(1.2.b)] if $N(x_{j+i+1}) \cap (A_{j+i} \cap \Red)\ne \emptyset$, then $C_i'=\cGG_{j+i}(C_{i+1}'\cup \{x_{j+i+1}\})$, else $C_i'=\cGG_{j+i}(C_{j+i+1}')$.
    \end{enumerate}
  \item[(1.3)] if $x_{j+i+1}\in \Red$, then either ($N(x_{j+i+1})\cap (C_i \cup C_{i+1}' )\ne \emptyset$) or ($(x_{j+i+1}\in \Blue$ and $|N(x_{j+i+1})\cap (R_i\cup R_{i+1}')|\in \rho^*$).
  \end{enumerate}

\item[(A2)] $R_p\equiv_{A_{j+p}}^d (R_{p-1}\cup \{x_{j+p}\})$, $R_{p-1}'\equiv_{\overline{A_{j+p-1}}}^d (R_p'\cup \{x_{j+p}\})$, $x_{j+p}\in \Red$, and
  \begin{enumerate}
  \item[(2.1)] if $x_{j+p}\in \Blue$, then $|N(x_{j+p})\cap (R_{p-1}\cup R_{p}')|\in \sigma$, 
  \item[(2.2)] if ($x_{j+p}\notin \Blue$ or  ($x_{j+p}\in \Blue$ and $|N(x_{j+p})\cap (R_{p-1}\cup R_{p}')|\notin \sigma^+$)), then  ($C_{p}=\cSG_{p}(C_{p-1})$ and $C_{p-1}'=\cGG_{p-1}(C_{p}')$), otherwise we
    should have ($|N(x_{j+p})\cap (R_{p-1}\cup R_{p}')|\in \sigma^+$) and
    \begin{enumerate}
    \item[(2.2.a)] if $N(x_{j+p})\cap (\comp{A}_{p}\cap \Red)\ne \emptyset$, then $C_{p}=\cSG_{p}(C_{p-1}\cup \{x_{j+p}\})$, else $C_{p}=\cSG_{p}(C_{p-1})$, and
    \item[(2.2.b)] if $N(x_{j+p}) \cap (A_{p-1} \cap \Red)\ne \emptyset$, then $C_{p-1}'=\cGG_{p-1}(C_{p}'\cup \{x_{j+p}\})$, else $C_{p-1}'=\cGG_{p-1}(C_{p}')$.
    \end{enumerate}
  \end{enumerate}
\end{enumerate}

\medskip

We now define arcs from the source node.  There is an arc from the source node to a node $(R,R',C,C',j)$ if
({\small $S=\{x\in (A_j\cap \Blue)\setminus \{x_j\}\mid N(x)\cap (\comp{A}_j\cap \Red)\ne \emptyset$ and $|N(x)\cap (\{x_j\}\cup R')|\in \rho^+\}$})
\begin{enumerate}[leftmargin=0.8cm]
\item[(S1)] $\{x_j\} \equiv_{A_j}^d R$ and $(\{x_j\} \cup R')$ $(\sigma,\rho)$-dominates $A_j\cap \Blue$,
\item[(S2)] if ($x_j\in \Blue$ and $|N(x_j)\cap R'|\in \sigma^+$) then $C=\cSG_j(S\cup \{x_j\})$, otherwise $C=\cSG_j(S)$, and
\item[(S3)] for each red vertex $x\in A_j\setminus \{x_j\}$, either ($N(x)\cap (C'\cup C)\ne \emptyset$) or ($x\in \Blue$ and $|N(x)\cap (R'\cup R)|\in \rho^*$).
\end{enumerate}

\medskip We finally define the arcs to the terminal node.  There is an arc from a node $(R,R',C,C',j)$ to the terminal node if
\begin{enumerate}[leftmargin=0.8cm]
\item[(T1)] $|N(x)\cap (R\cup R')|\in \rho$ for each $x\in \comp{A}_{j+1}\cap \Blue$, and
\item[(T2)] $C'=\cGG_j(\{x\in \comp{A}_j\cap \Blue\mid N(x)\cap (A_j\cap \Red)\ne \emptyset\ \text{and}\ |N(x)\cap R|\in \rho^+\})$,
\item[(T3)] for each red vertex in $\comp{A}_{j}$, then $N(x)\cap (C\cup C')\ne \emptyset$ or ($|N(x) \cap (R\cup R')|\in \rho^*$ if $x\in \Blue$).
\end{enumerate}

\bigskip

$DAGM(G)$ is clearly a DAG, and as for $DAG(G)$ one can construct it in time $O(n^{c\cdot d})$.  Similarly, one can observe that if $P:=(s,v_1,\ldots,v_k,t)$ is a maximal path from the source node to
the terminal node, then if we let $D_i:=\{x_{j_1},\ldots,x_{j_i}\}$ with $v_i:=(R_{j_i},R_{j_i}',C_{j_i},C_{j_i}',j_i))$, then

\begin{itemize}
\item[(i)] $R_{j_i}\in LR_{j_i}$, $R_{j_i}'\in LR_{\bar{j_i}}$;
\item[(ii)]  $D_i\equiv_{A_{j_i}}^d R_{j_i}$, 
\item[(iii)] $D_i\cup R_{j_i}'$ $(\sigma,\rho)$-dominates $A_{j_i}\cap \Blue$;
\item[(iv)] Each $u\in A_{j_i}\setminus D_i$ is either adjacent to a vertex from $C_{j_i}'$, or has a certificate in $A_{j_i}\cap \Blue$.
\item[(v)]  $C_{j_i}=\cSG_{j_i}(S_i)$, where $S_i$ is the set of vertices in $A_{j_i}\cap \Blue$ that are certificates and have a neighbor in $\comp{A}_{j_i}\cap \Red$;
\item[(vi)] $C_{j_i}'=\cGG_{j_i}(S_i')$ where $S_i'$ is the set of vertices in $\comp{A}_{j_i}\cap \Blue$ that are certificates and have a neighbor in
  $A_{j_i}\cap \Red$.
\end{itemize}

Hence, we can prove counterparts to Lemmas \ref{lem:path-dom} and \ref{lem:dom-path} and deduce the following theorem from Section \ref{subsec:mindag-defn}.

\begin{thm}\label{thm:main-ext-max} The set of $1$-maximal \Red $(\sigma,\rho)$-dominating sets in $G$ can be enumerated with linear delay and with polynomial space.  We can moreover count in time $O(n^{c\cdot d})$ the number of $1$-maximal \Red $(\sigma,\rho)$-dominating sets in $G$.
\end{thm}

\section{Enumeration of minimal dominating sets for unit square graphs}\label{sec:unit-square}
In this section we prove that all minimal dominating sets of a unit square graph can be enumerated in incremental polynomial time. In Section~\ref{sec:local-mim} we show that the class of unit square graphs has locally bounded LMIM-width. In Section~\ref{sec:flipping} we use this property and Theorem~\ref{thm:mainthm}, to obtain an enumeration algorithm for minimal dominating sets. 
To do it, we use the \emph{flipping} method proposed by Golovach,  Heggernes,  Kratsch and Villanger in~\cite{GolovachHKV14}.

\subsection{Local LMIM-width of unit square graphs}\label{sec:local-mim}
First, we introduce some additional notations. For $x,y\in\mathbb{R}$ such that $x\leq y$, $[x,y]=\{z\in\mathbb{R}\mid x\leq z\leq y\}$. 
Let $G$ be a unit square graph and suppose that $f\colon V(G)\rightarrow\mathbb{Q}^2$ is a realization of $G$. (See Section~\ref{sec:defs} for more details on the point model of unit
square graphs used in our paper.) 
For a vertex $v\in V(G)$, 
$\fr(v)=x_f(x)-\lfloor x_f(v)\rfloor$ is the fractional part of the $x$-coordinate of the point representing $v$.

\begin{lem}\label{lem:matching}
Let $G$ be a unit square graph with a realization $f$ such that for every $v\in V(G)$ the point $f(v)$ belongs to $[1,w] \times [1,h]$, where $h,w \in\bN$. If for $x\in [1,w]$, $A=\{v\in V(G)\mid x_f=x\}$ is a non-empty proper subset of $V(G)$, then $\mim_G(A,\overline{A})\leq h$.
\end{lem}

\begin{proof}
  Let $M$ be a maximum induced matching in $G[A,\comp{A}]$. Denote by
$M_A$ the set of end-vertices of the edges of $M$ in $A$ and let $M_B$ be the set of end-vertices of the matching in $\comp{A}$.
As all the vertices of $M_A$ have the same $x$-coordinate, and each vertex of
$M_B$ is adjacent to some vertex in $M_A$, a vertex $u \in M_A$ is adjacent to a
vertex $v \in M_B$  only if 
$|y_f(u) - y_f(v)| < 1$. Denote by $a_1,\ldots,a_k$ and $b_1,\ldots,b_k$ the vertices of $M_A$ and $M_B$ respectively 
and assume that they are ordered by the increase of their $y$-coordinates.

We claim that $a_ib_i\in M$ for $i\in\{1,\ldots,k\}$. To obtain a contradiction, suppose that there is $a_i$ that is not adjacent to $b_i$ and choose the minimum index $i$ for which it holds.
Then $a_ib_j\in M$ and $b_ia_s\in M$ for some $j,s>i$.  
If $b_{i}$ is adjacent to $a_{s}$ but not $a_i$, we must have $|y_f(b_i)-y_f(a_i)|\geq 1$ and $|y_f(b_i)-y_f(a_s)|<1$.
Since $y_f(a_s)\geq y_f(a_i)$, $y_f(b_i)\geq y_f(a_i)+1$.
 But as $y_f(b_j) \geq y_f(b_i) \geq y_f(a_i) + 1$, $b_j$ and $a_i$ cannot be
adjacent after all; a contradiction.

Now we show that $y_f(a_i)\geq y_f(b_{i-1}) + 1$ and $y_f(b_i) \geq y_f(a_{i-1}) + 1$ for $i\in\{2,\ldots,k\}$. 
Because $a_{i-1}b_{i-1}\in M$, $|y_f(a_{i-1})-y_f(b_{i-1})|< 1$. As $a_{i-1}b_i,a_ib_{i-1}\notin E(G)$, $|y_f(a_{i-1})-y_f(b_{i})|\geq 1$ and $|y_f(a_{i})-y_f(b_{i-1})|\geq1$.
We have that  $y_f(a_i)\geq y_f(b_{i-1}) + 1$ and $y_f(b_i) \geq y_f(a_{i-1}) + 1$, because $y_f(a_i)\geq y_f(a_{i-1})$ and $y_f(b_i)\geq y_f(b_{i-1})$.

Next, we claim that $y_f(a_i),y_f(b_i)\geq i$. Clearly, $y_f(a_1),y_f(b_1)\geq 1$. 
Because $y_f(a_i)\geq y_f(b_{i-1}) + 1$ and $y_f(b_i) \geq y_f(a_{i-1}) + 1$ for $i\in\{2,\ldots,k\}$, we have that claim holds for all $i\in\{1,\ldots,k\}$
by induction.

Because $k\leq y_f(a_k)\leq h$, we conclude that $\mim_G(A,\comp{A})=k\leq h$.
\end{proof}

\begin{lem}\label{lem:loc-mim-w}
Let $G$ be a unit square graph with a realization $f$ such that for every $v\in V(G)$ the point $f(v)$ belongs to  $[1,w] \times [1,h]$, where $h,w \in\bN$. 
Then $\lmimw(G)\leq 2hw$. Moreover, a linear ordering of vertices of MIM-width at most $2hw$ can be constructed in polynomial time.
\end{lem}

\begin{proof}
Let $v_1,\ldots,v_n$ be the vertices of $G$ ordered by increasing $\fr$-value, i.e., $\fr(v_i)\leq \fr(v_j)$ if $i\leq j$. We show that this is an ordering of MIM-width at most $2hw$.

For contradiction, assume that for some $i\in\{1,\ldots,n-1\}$,  $\mim_G(A,\comp{A})\geq 2hw+1$ where $A=\{v_1,\ldots,v_i\}$, i.e., 
the graph $G[A, \comp{A}]$ has an induced matching $M$ of size $2hw+1$. 
Let $V(M)$ denote the set of the end-vertices of the edges of $M$. By the pigeonhole principle, for
some positive integer $p\leq w$ there are at least $2h + 1$ vertices $v\in V(M) \cap A$ so that $\lfloor x_f(v)\rfloor = p$. 
Denote this subset of $A \cap
V(M)$ by $C_A$, and denote by $C_B$ the vertices of $V(M) \setminus A$ adjacent
to $C_A$.  Let $t = \max\{\fr(v) \mid v \in A\}$, and observe that $\{ v \mid \fr(v) < t \} \subset A$ and $\{ v \mid \fr(v)> t \} \cap A = \emptyset$.
We now partition $C_A$ into two parts $C_A^{<t} = \{v \in C_A \mid \fr(v) < t\}$
and $C_A^{=t} = \{v \in C_A \mid \fr(v) = t \}$ and argue that neither of these
parts can be of size more than $h$, contradicting that $|C_A| \geq 2h+1$.

We first show that $|C_A^{=t}| \leq h$. For each $v\in C_A^{=t}$,  
$\lfloor x_f(v)\rfloor  = p$ and $\fr(v) = t$. Hence,  $x_f(v) = p + t$ for all $v \in C_A^{=t}$.
By Lemma~\ref{lem:matching}, the size of the maximum induced matching in  $G[C_A^{=t},C_B]$ is at most $h$ and this implies that 
$|C_A^{=t}|\leq h$.

To show that also $|C_A^{<t}| \leq h$, we will show that for the sake of the
induced matching, all the $x$-coordinates of the vertices $v$ of $C_A^{<t}$
might as well have $\fr(v) = 0$, and therefore we can apply Lemma~\ref{lem:matching}.

Let $v\in C_A^{<t}$. Then we construct a new vertex $v'$ represented by the point $(\lfloor x_f(v)\rfloor,y_f(v))$.
We will now show that $v'$  is adjacent to a vertex $u \in C_B$ if and only if $v$ is adjacent to
$u$. As the $y$-coordinates of $v$ and $v'$ are the same, we only
need to prove that $|x_f(v) - x_f(u)| < 1$ if and only if $|\lfloor x_f(v)\rfloor -x_f(u)|< 1$. 

Suppose that  $v$ is adjacent to $u$ but $v'$ is not. Because $x_f(v)\geq \lfloor x_f(v)\rfloor$, we have that
$\lfloor x_f(v)\rfloor+1\geq  x_f(u)>x_f(v) + 1$. 
However, that means $\fr(u) \leq \fr(v) < t$, which
implies that $u\in A$ contradicting $u\in C_B$. Similarly,
suppose $v'$ is adjacent to $u$ but $v$ is not. Now 
$\lfloor x_f(v)\rfloor-1< x_f(u)\leq x_f(v)-1$,  which again implies
that $\fr(u) \leq \fr(v) < t$, contradicting that $u$ is in $C_B$.

Consider $S = \{v' \mid v \in C_A^{<t} \}$, where each $v'$ is represented by the point $(\lfloor x_f(v)\rfloor,y_f(v))=(p,y_f(v))$. 
Because each $v'\in S$ is adjacent to $u\in C_B$ if and only if $v$ is adjacent to $u$, by Lemma~\ref{lem:matching}, 
$|C_A^{<t}|=|S|\leq h$.

It remains to show that the ordering of $V(G)$ can be constructed in polynomial time. Clearly, the ordering can be done in time $O(n\log n)$ if we assume that we can compute $\fr(v)$ and compare the $\fr$-values of two vertices in time $O(1)$. Otherwise, if the table of the values $f\colon V(G)\rightarrow \mathbb{Q}^2$ is given in the input, we still can produce the ordering in polynomial time.
\end{proof}

Now we are ready to show that the class of unit square graphs has locally bounded LMIM-width.

\begin{thm}\label{thm:loc-bound}
For a unit square graph $G$, $u\in V(G)$ and a positive integer $r$, $\lmimw(G[N_G^r[u]])=O(r^2)$.
Moreover, if a realization $f\colon V(G)\rightarrow\mathbb{Q}^2$ of $G$ is given, then  a linear ordering of the vertices of MIM-width $O(r^2)$ can be constructed in polynomial time.
\end{thm}

\begin{proof}
Let $f$ be a realization of $G$. Without loss of generality we may assume that $\min\{x_f(v)\mid v\in N_G^r[u]\}=\min\{y_f(v)\mid v\in N_G^r[u]\}=1$. Otherwise, we can shift the points representing vertices.
For any $v\in N_G^r[u]$, $|x_f(u)-x_f(v)|\leq r$ and $|y_f(u)-y_f(v)|\leq r$. We obtain that $f(v)\in[1,2r+1]\times[1,2r+1]$. By Lemma~\ref{lem:loc-mim-w},
 $\lmimw(G[N_G^r[u]])=O(r^2)$ and the corresponding ordering of the vertices can be constructed in polynomial time.
\end{proof}

\subsection{Enumeration by flipping for graphs of locally bounded LMIM-width}\label{sec:flipping}
We use a variant of the \emph{flipping} method proposed by Golovach,  Heggernes,  Kratsch and Villanger in~\cite{GolovachHKV14}.
Given a minimal dominating set $D^*$, the flipping operation replaces an isolated vertex of $G[D^*]$ with its neighbor outside of $D^*$, and, if necessary, adds or deletes some vertices to obtain new minimal dominating sets $D$, such that $G[D]$ has more edges compared to $G[D^*]$. The enumeration algorithm starts with enumerating all maximal independent sets of the input graph $G$ using the algorithm of Johnson,  Papadimitriou, and Yannakakis \cite{JohnsonP88}, which gives the initial minimal dominating sets. Then the flipping operation is applied to every appropriate minimal dominating set found, to find new minimal dominating sets inducing subgraphs with more edges. 

Let $G$ be a graph. Let also $D\subseteq V(G)$. For $u\in D$, $C_D[u]=\{v\in V(G)\mid v\in N_G[u]\setminus N_G[D\setminus\{v\}]\}$ and 
$C_D(u)=\{v\in V(G)\mid v\in N_G(u)\setminus N_G[D\setminus\{v\}]\}=C_D[u]\setminus\{u\}$. Observe that if $D$ is a minimal dominating set, then $C_D(u)$ is the set of certificates for a vertex $u\in D$.

Let us describe the variant of the flipping operation from  \cite{GolovachHKV14}, that we use. Let $G$ be the input graph; we fix an (arbitrary) order of its vertices: $v_1,\ldots,v_n$. 
Suppose that $D'$ is a dominating set of $G$. We say that the minimal dominating set $D$ is obtained from $D'$ by \emph{greedy removal of vertices} ({\em with respect to order $v_1,\ldots,v_n$}) if 
we initially let $D=D'$, and then recursively apply the following rule:
{\em If $D$ is not minimal, then find a vertex $v_i$ with the smallest index $i$ such that
 $D\setminus\{v_i\}$ is a dominating set in $G$, and set $D=D\setminus\{v_i\}$.}
Clearly, when we apply 
this rule,  we never remove vertices of $D'$ that have certificates.  Whenever greedy removal of vertices of a dominating set is performed,  it is done with respect to this ordering.  

Let $D$ be a minimal dominating set of $G$ such that $G[D]$ has at least one edge $uw$. Then the vertex $u\in D$ is dominated by the vertex $w\in D$. Therefore, $C_D[u]=C_D(u)\neq \emptyset$.
Let $X$ be an non-empty inclusion-maximal independent set such that  $X\subseteq C_D(u)$. 
Consider the set $D'=(D\setminus\{u\})\cup X$. Notice that $D'$ is a dominating set in $G$, since all vertices of $C_D(u)$ are dominated by $X$ by the maximality of $X$ and $u$ is dominated by $w$, 
but $D'$ is not necessarily minimal, because it can happen that  $X$ dominates all the certificates of some vertex of $D\setminus \{u\}$.
We apply greedy removal of vertices to $D'$ to obtain a minimal dominating set. 
Let $Z$ be the set of vertices that are removed by this to ensure minimality. 
Observe that $X \cap Z = \emptyset$ and $u\notin Z$ by the definition of these sets; in fact there is no edge between a vertex of $X$ and a vertex of $Z$. Finally, let   $D^*=((D\setminus \{u\})\cup X)\setminus Z$.

It is important to notice that $|E(G[D^*])|<|E(G[D])|$. Indeed, to construct $D^*$, we remove the endpoint $u$ of the edge $uw\in E(G[D])$ and, therefore, reduce the number of edges. Then we add $X$ but these vertices form an independent set in $G$ and, because they are certificates for $u$ with respect to $D$, they are not adjacent to any vertex of $D\setminus\{u\}$. Therefore, $|E(G[D^*])|\leq|E(G[D'])|<|E(G[D])|$. 

The {\em flipping} operation is exactly the {\em reverse} of how we generated $D^*$ from $D$; \ie, it replaces a non-empty independent set $X$ in $G[D^*]$ such that  $X\subseteq G[D^*]\cap N_G(u)$ for a vertex $u\notin D^*$ with their neighbor $u$ in $G$ to obtain $D$.  In particular, we are interested in all minimal dominating sets $D$ that can be generated from $D^*$ in this way. 
Given $D$ and $D^*$ as defined above, we say that $D^*$ is a \emph{parent of $D$ with respect to flipping $u$ and $X$}. We say that $D^*$ is a \emph{parent} of $D$ if there is a  vertex $u\in V(G)$ and an independent set $X\subseteq N_G(u)$ such that 
$D^*$ is a parent with respect to flipping $u$ and $X$. It is important to note that each minimal dominating set $D$ such that $E(G[D])\neq\emptyset$ has a unique parent with respect to flipping of any $u\in D\cap N_G[D\setminus\{u\}]$ and a maximal independent set $X\subseteq C_D(u)$, as $Z$ is lexicographically first sets selected by a greedy algorithm. 
Similarly, we say that $D$ is a \emph{child} of $D^*$ (with respect to flipping $u$ and $X$) if $D^*$ is the parent of $D$ (with respect to flipping $u$ and $X$).

The proof of the following lemma is implicit in~\cite{GolovachHKV14}.

\begin{lem}[\cite{GolovachHKV14}]\label{lem:main}
Suppose that for a graph $G$, all independent sets $X\subseteq N_G(u)$ for a vertex $u$ can be enumerated in polynomial time.
Suppose also that there is an enumeration algorithm $\mathcal{A}$ that, given a minimal dominating set $D^*$ of a graph $G$
 such that $G[D^*]$ has an isolated vertex, a vertex $u\in V(G)\setminus D^*$ and a non-empty independent set $X$ of $G[D^*]$ such that
$X\subseteq D^*\cap N_G(u)$, generates with polynomial delay a family of minimal dominating sets $\mathcal{D}$ with the property that $\mathcal{D}$ contains all minimal dominating sets $D$ 
that are children of $D^*$ with respect to flipping $u$ and $X$.
Then all minimal dominating sets of $G$ can be enumerated in incremental polynomial time.
\end{lem}

To obtain our main result, we will show that there is indeed an algorithm as algorithm $\mathcal{A}$ described in the statement of Lemma \ref{lem:main} when the input graph $G$ is a unit square graph.
We
show that we can construct $\mathcal{A}$ by reduction to the enumeration of minimal \Red dominating set in an auxiliary colored induced subgraph of $G[N_G^3[u]]$.  
Let $D^*$ be a minimal dominating set of a graph $G$ such that $G[D^*]$ has an isolated vertex. Let also $u\in V(G)\setminus D^*$ and $X$ is a non-empty independent set  of $G[D^*]$ such that $X\subseteq D^*\cap N_G(u)$. Consider the set $D\rq{}=(D\setminus X)\cup \{u\}$. Denote by $\Blue$ the set of vertices that are not dominated by $D\rq{}$. Notice that $\Blue\subseteq N_G(X)\setminus N_G[u]$. Therefore, $\Blue\subseteq N_G^2[u]$.  Let $\Red=N_G(\Blue)\setminus N_G[X]$. Clearly, $\Red\subseteq N_G^3[u]$.
We construct the colored graph $H=G[\Red\cup\Blue]$. Let $\mathcal{A}'$ be an algorithm that enumerates minimal \Red dominating sets in $H$. Assume that if $\Blue=\emptyset$, then $\mathcal{A}'$ returns $\emptyset$ as the unique \Red dominating set. We construct $\mathcal{A}$ as follows.
\begin{itemize}[leftmargin=1.2cm]
\item[Step 1.] If $\mathcal{A}'$ returns an empty list of sets, then $\mathcal{A}$ returns an empty list as well. 
\item[Step 2.] For each \Red dominating set $R$ of $H$, consider $D''=D'\cup R$ and construct a minimal dominating set $D$ from $D''$ by greedy removal.
\end{itemize}

\begin{lem}\label{lem:corr-enum}
If $\mathcal{A}'$ lists all minimal \Red dominating sets with polynomial delay, then $\mathcal{A}$ generates with polynomial delay a family of minimal dominating sets $\mathcal{D}$ with the property that $\mathcal{D}$ contains all minimal dominating sets $D$ 
that are children of $D^*$ with respect to flipping $u$ and $X$.
\end{lem}

\begin{proof}
First, we show that $\mathcal{A}$ produces pairwise distinct minimal dominating sets of $G$.
Let $R$ be a \Red dominating set of $H$. The set $D'=(D\setminus X)\cup \{u\}$ dominates all vertices of $G$ except the vertices of $\Blue$. Since $R$ dominates $\Blue$, $D''$ is a dominating set of $G$ and, therefore, $D^*$ obtained from $D''$ by the greedy removal is a minimal dominating set. 
To see that all generated sets are distinct, observe that every vertex of $R$ has its certificate in $B$. Therefore, the greedy removal never deletes vertices of $R$. Since all sets $R$ generated by $\mathcal{A}'$
are pairwise distinct, the claim follows.

Let $D$ be a child of $D^*$ with respect to flipping $u$ and $X$. Then $D=((D^*\cup\{u\})\setminus X)\cup Z$. Recall that $u\notin Z$, $X \cap Z = \emptyset$ and the vertices of $Z$ are not adjacent to the vertices of $X$  by the definition of $X$ and $Z$. Hence, $Z\cap \Blue=\emptyset$.  
Because $D$ is minimal, each vertex of $Z$ has a certificate. As only the vertices of $\Blue$ are not dominated by $(D^*\cup\{u\})\setminus X$, each vertex of $Z$ has its certificate in $\Blue$.
It remains to observe that $Z$ dominates $\Blue$, to see that $Z$ is a minimal \Red dominating set of $H$.  Because $\mathcal{A}'$ generates all minimal \Red dominating sets, we have that $D\in\mathcal{D}$.
\end{proof}

Now we are ready to prove the main result of the section.

\begin{thm}\label{thm:enum-unit-square}
For a unit square graph $G$ given with its realization $f$, all minimal dominating sets of $G$ can be enumerated in incremental polynomial time.
\end{thm}

\begin{proof}
It is straightforward to observe that for a vertex $u$ of a unit square graph $G$, any independent set $X\subseteq N_G(u)$ has at most 4 vertices. Hence, all independent sets $X\subseteq N_G(u)$ for a vertex $u$ can be enumerated in polynomial time.
By combining 
Theorems~\ref{thm:mainthm} and \ref{thm:loc-bound}, and Lemmas~\ref{lem:main} and \ref{lem:corr-enum}, we obtain the claim. 
\end{proof}

\end{document}